\documentclass[a4paper,12pt]{article}
\usepackage{amsmath,amssymb,amsfonts,amsthm}
\newtheorem{theorem}{Theorem}[section]
\newtheorem{lemma}[theorem]{Lemma}

\newcommand{\gc}{g}
\newcommand{\nc}{\bar \nabla}
\newcommand{\vc}{\mathcal{V}}

\newcommand{\rc}{^{(4)}\mathcal R}

\newcommand{\norm}{\lambda}

\newcommand{\Y}{\omega}

\newcommand{\Dt}{\bar D}

\newcommand{\qt}{\gamma}
 
\newcommand{\et}{K}

\newcommand{\slicet}{\bar S}
\newcommand{\sliced}{S}

\newcommand{\gt}{h}

\newcommand{\nt}{\nabla}
\newcommand{\vt}{\mathcal{N}}

\newcommand{\rtl}{^{(3)}\mathcal R}

\newcommand{\qd}{q}

\newcommand{\Dd}{D}

\newcommand{\ed}{\chi}

\newcommand{\sd}{\beta}

\newcommand{\ld}{\alpha}

\newcommand{\Nt}{n}

\newcommand{\rd}{{^{(2)}R}}

\newcommand{\rt}{{^{(3)}R}}

\newcommand{\dvq}{\, dV_q}

\newcommand{\Rt}{\mathbb{R}^3}
\newcommand{\Rd}{\mathbb{R}^2}

\newcommand{\ckq}{\mathcal{L}_q}
\newcommand{\ck}{\mathcal{L}}

\newcommand{\Ld}{\Delta}

\newcommand{\Ldt}{{^{(3)}\Delta}}

\newcommand{\Lg}{\Delta_{\gamma}}

\newcommand{\q}{q}
\newcommand{\x}{\sigma}

\renewcommand{\u}{u}

\newcommand{\hkb}[2]{H^{#1}_{#2}}

\newcommand{\oi}[1]{o_\infty(r^{#1})}

\title{Axisymmetric evolution of Einstein equations and
  mass conservation}

\author{Sergio Dain\\
  Facultad de Matem\'atica, Astronom\'{i}a y F\'{i}sica\\
  Universidad Nacional de C\'ordoba\\
  Ciudad Universitaria \\
  (5000)  C\'ordoba\\
  Argentina.\\
  \\
  Max Planck Institute for Gravitational Physics\\
  (Albert Einstein Institute)\\
  Am M\"uhlenberg 1\\
  D-14476 Potsdam\\
  Germany.}
\begin{document}
\maketitle

\begin{abstract}
  For axisymmetric evolution of isolated systems, we prove that there
  exists a gauge such that the total mass can be written as a positive
  definite integral on the spacelike hypersurfaces of the foliation
  and the integral is constant along the evolution. The conserved mass
  integral controls the square of the extrinsic curvature and the
  square of first derivatives of the intrinsic metric.  We also
  discuss  applications of this result for the global
  existence problem in axial symmetry.
\end{abstract}

\section{Introduction}
\label{sec:introduction}

The main unsolved problem in General Relativity is the global
existence of solutions of Einstein's equations that describe the
dynamics of strong gravitational fields.  Symmetries have
traditionally play an important role in this problem. The presence of
a symmetry reduces the degrees of freedom of Einstein equations and
hence simplify considerable its analysis. This is of course useful as
a preliminary step to understand the full problem but also many models
with symmetries have direct physical applications.

In vacuum, due to Birkhoff's theorem, spherical symmetry has no
dynamics. For isolated systems, the next possible model with
symmetries are axially symmetric spacetimes. It has been
proved in \cite{bicak98} that no additional symmetry can be imposed to
the spacetime if we want to keep the gravitational radiation and a
complete null infinity. This result single out axially symmetric
spacetimes as the only models for isolated, dynamical, system with
symmetries. 

There exists many relevant physical models one can study in axial
symmetry: head-on collisions of two black holes, rotating starts and
black holes, critical collapse of gravitational waves. These models
have been studied numerically.  Particularly relevant for the results
presented in this article are the following references in which axial
symmetry is imposed explicitly on the equations using cylindrical
coordinates: \cite{Nakamura87}, \cite{Bardeen83},
\cite{Garfinkle:2000hd}, \cite{Choptuik:2003as}, \cite{Rinne:2005df},
\cite{Rinne:2005sk}, \cite{Ruiz:2007rs}, \cite{Rinne:2008tk}.

From the analytical point of view, there exists up to now no results
for axially symmetric isolated systems (see the review
\cite{Rendall:lrr-2005-6} for results with other kind of symmetries in
cosmologies). One of the purpose of this article is to initiate the
study of this problem. We mention a related
problem: cosmologies with $U(1)$ symmetries. This symmetry was
analyzed in \cite{Choquet-Bruhat01}, \cite{Choquet-Bruhat04}.  The
equations are locally the same as in the axially symmetric case but
the boundary conditions are radically different.  In axial symmetry
the Killing vector vanishes at the axis and that is the main source of
difficulties.

For axially symmetric data, the total ADM mass \cite{Arnowitt62} can
be written as a positive definite volume integral over one spacelike
hypersurface (see \cite{Brill59}, \cite{Gibbons06}, \cite{Dain06c},
\cite{Moncrief:08}, \cite{Chrusciel:2007dd}). This fact is likely to
play a major role in the initial value (see the discussion in
\cite{Dain:2007pk}).  In order to study the implication of this
integral formula for the evolution, the first natural step is to prove
that there exists a gauge such that the mass integral holds not only
at one hypersurface but in a whole foliation and it is conserved along
the evolution. This is the subject of the present article.  To
understand this result, let us review the notion of mass in General
Relativity and its conservation.


The total mass of an isolated system is a boundary integral at
infinity calculated on a given spacelike hypersurface. If we consider
not only one hypersurface but a foliation on the spacetime we can ask
the question whether the mass is conserved along this particular
foliation. That is, whether the boundary integral gives the same
result if it calculated on different slices of the foliation. Any
foliation is determined by the choice of a lapse function and a shift
vector.  If the lapse and shift of the foliation satisfy some fall-off
conditions (this class is called asymptotically flat
gauges), then the mass boundary integral is conserved. This notions
appears naturally in the Hamiltonian formulation of General Relativity
(see \cite{Regge74} \cite{Beig87} \cite{Szabados03}\cite{Szabados06}).

In any physical theory conserved quantities (in particular, conserved
energies) are very important to control the evolution of the
system. However, in General Relativity, the conserved mass appears as
a boundary integral and not as a volume integral (as, for example, in
the wave equation). Hence it is not possible to relate the mass with
any norm of the fields to control the evolution of them (for the wave
equation the energy is precisely the norm of the wave). Axially
symmetric systems represent a remarkable exception. As it was
mentioned above in this case it is possible to write the mass as a
positive definite volume integral. However, this integral formula is
valid only in a particular gauge. Moreover, this gauge can not be
given a priori, it is a dynamical gauge which is prescribed as a
solution of a system of differential equations. The natural question
is if this particular gauge satisfy the fall-off conditions which
guarantee the conservation of the mass.  Our main result is that the
answer to that question is yes. In other words, there exist gauge for
which the mass can be written as a positive definite integral on each
slice of the foliation and this integral is conserved along the
evolution on this particular foliation. Then, in this gauge, the
conserved mass controls the norm of the metric and the second
fundamental form along the evolution and hence this particular gauge
is likely to be the most relevant one to study axially symmetric
isolated systems (we further discuss this point  in section
\ref{sec:final-comments}).

The plan of the article is the following. In section
\ref{sec:main-result} we describe our main result which is given by
theorem \ref{teo:main-result-1}. In section
\ref{sec:axisymm-evol-equat} we review the well known (2+1)+1
formalism and compute the behavior of the fields at infinity and near
the axis. In section \ref{sec:maxim-isoth-gauge} we describe the
gauge. We also derive the mass integral formula in the (2+1)+1
formalism. This derivation is remarkably simpler than the ones
presented so far in the literature. In section
\ref{sec:existence-gauge-mass} we prove the main result. Finally, in
section \ref{sec:final-comments} we discuss the implication of this
result for the evolution problem.
\section{Main result}
\label{sec:main-result}

An \emph{initial data set} for the Einstein vacuum equations is given
by a triple $(\slicet, \qt_{ab}, \et_{ab})$ where $\slicet$ is a connected
3-dimensional manifold, $ \qt_{ab} $ a (positive definite) Riemannian
metric, and $ \et_{ab}$ a symmetric tensor field on $\slicet$, such that the
vacuum constraint equations
\begin{align}
 \label{const1}
   \Dt_b   K^{ab} -  \Dt^a   K= 0,\\
 \label{const2}
   \rt -  K_{ab}   K^{ab}+  K^2=0,
\end{align}
are satisfied on $\slicet$. Where $\Dt$ and $\rt$ are the
Levi-Civita connection and the Ricci scalar associated with
$ \qt_{ab}$, and $  \et =   \et_{ab}   \qt^{ab}$. In
these equations the indices are moved with the metric $ \qt_{ab}$
and its inverse $  \qt^{ab}$. For simplicity, in the present article we
assume  $\slicet=\Rt$. We expect, however, that the results presented
here generalize to asymptotically euclidean manifold with many ends.

We will further assume that the data are \emph{axially symmetric},
which means that there exists a Killing vector field $\eta^a$, i.e;
\begin{equation}
  \label{eq:38}
 \pounds_\eta   \qt_{ab}=0,
\end{equation}
where $\pounds$ denotes the Lie derivative, 
which  has complete periodic orbits and such that
\begin{equation} 
  \label{eq:8c} 
 \pounds_\eta   \et_{ab}=0.
 \end{equation} 
The Killing vector field has associated the following scalars with
respect to the metric $\qt_{ab}$. The norm 
\begin{equation}
  \label{eq:52}
  \norm^2 = \qt_{ab}\eta^a\eta^b,
\end{equation}
and the twist  
\begin{equation}
  \label{eq:55}
  \Y'=\epsilon_{abc}\eta^a\Dt^b\eta^c.
\end{equation}
The prime in the notation of $\Y'$ is justified in section
\ref{sec:axisymm-evol-equat}, where we will show that $\Y'$ is
related to the time derivative of the four dimensional twist potential
$\Y$.

The data are called \emph{asymptotically flat} if there exists a
coordinate system $x^\mu$ ($\mu=1,2,3$) in $\Rt$ such that the metric
components in these coordinates satisfy the fall-off conditions
\begin{equation}
  \label{eq:66}
  \qt_{\mu\nu}= \delta_{\mu\nu} + \oi{-1/2},
\end{equation}
and the extrinsic curvature components
\begin{equation}
  \label{eq:76}
   \quad  \et_{\mu\nu} = o_\infty(r^{-3/2}),
\end{equation}
where we write $f=o_\infty(r^{k})$ if $f$ satisfies $\partial^\alpha
f=o(r^{k-|\alpha|})$, for all $\alpha$. Here $\delta_{\mu\nu}$ is the
euclidean metric in the coordinates $x^\mu$,  $r$ is the euclidean
radius $r=\sqrt{x^\mu x^\nu\delta_{\mu\nu}}$, $\partial $ denotes
partial derivatives with respect to the coordinates and $\alpha$ is a
multindex. 

The data are called \emph{maximal} if the trace of $\et_{ab}$ is equal to
zero, i.e. 
\begin{equation}
  \label{eq:56}
  K=\qt^{ab}\et_{ab}=0.
\end{equation}

Following \cite{Geroch71}, we define  the  two dimensional quotient metric
$\qd_{ab}$  by
\begin{equation}
  \label{eq:51}
  \qd_{ab}= \qt_{ab}-\frac{\eta_a\eta_b}{\norm^2}.
\end{equation}
In \cite{Chrusciel:2007dd} it is proved that for any axially symmetric
and asymptotically flat metric $\qt_{ab}$ there exists a coordinate
system $(\rho,z, \phi)$, which we call the \emph{isothermal
  coordinates}, such that the metric $\qd_{ab}$ and the Killing vector
have the following form
\begin{equation}
  \label{eq:60}
  q =e^{2\u } (d\rho^2+dz^2),\quad \eta^a=\left( \frac{\partial}{\partial
      \phi}\right)^a,  
\end{equation}
and the norm is given  by  
\begin{equation}
  \label{eq:61}
  \norm=\rho e^{\x/2},
\end{equation}
where $\u$ and $\x$ are smooth functions. To these  coordinates we
associated Cartesian coordinates $(x,y,z)$ via the standard formula
$x=\rho \cos \phi $, $y=\rho \sin \phi $. It follows from the results
proved in \cite{Chrusciel:2007dd} that in the  coordinates $(x,y,z)$
the metric $\qt_{ab}$ is asymptotically flat. Hence, without lose of
generality, in the following we take  $x^\mu$ to be   $(x,y,z)$.

In isothermal coordinates we consider the following integral
\begin{equation}
  \label{eq:15}
m= \frac{1}{16}\int_{-\infty}^\infty dz \int_{0}^\infty \left[|\partial \x |^2
+ \frac{e^{2\u} \Y'^2}{\norm^4}  + 2e^{2\u}\et^{ab}\et_{ab}
\right ] \,\rho \, d\rho, 
\end{equation}
where $|\partial \x |^2= \x^2_{,\rho} +\x^2_{,z}$.
Note that the volume element $\rho dz d\rho$ is equivalent to the
volume element in Cartesian coordinates $dxdydz$ where the integration
domain is $\Rt$. We will see in section \ref{sec:maxim-isoth-gauge}
that for asymptotically flat, axially symmetric and
maximal data $m$ defined by (\ref{eq:15}) is the total mass.

The mass integral (\ref{eq:15}) is positive definite, and is zero only
for flat spacetime. The first two terms in the integrand contain the
square of first derivatives of the metric and the third term contains  the
square of the extrinsic curvature.

We want to prescribe gauge conditions for the evolution such that the
integral (\ref{eq:15}) holds not only at the initial surface but in
any surface of the foliation. We prescribe the lapse by the requirement
that the maximal condition \eqref{eq:56} is preserved under
evolution. The shift vector is prescribed by the condition that the
isothermal coordinates are preserved by the evolution. In this way we
obtain a coordinate system $(t,\rho,z,\phi)$, we call it the \emph{maximal
  isothermal gauge}.  These gauge conditions are explicitly given in
section \ref{sec:maxim-isoth-gauge}. This gauge has been used in
numerical evolutions in the following references \cite{Nakamura87},
\cite{Garfinkle:2000hd}, \cite{Choptuik:2003as}, \cite{Rinne:2005df},
\cite{Rinne:2008tk}. This is a dynamical gauge, namely, the lapse and
shift are calculated in terms of the initial data at each step of the
evolution.

The main result of this article is given by the following theorem.
\begin{theorem}
\label{teo:main-result-1} 
  Let $(\Rt,\qt_{ab}, \et_{ab}, \eta^a)$ be an axially symmetric, maximal,
  vacuum, asymptotically flat initial data. Then, there exist a unique
  solution of the maximal-isothermal gauge equations and 
  the mass integral given by \eqref{eq:15} is conserved along the
  evolution, that is
\begin{equation}
  \label{eq:1b}
  \frac{dm(t)}{dt}=0.
\end{equation}
\end{theorem}
This theorem is proved is section \ref{sec:maxim-isoth-gauge}. We
divide the proof in lemma \ref{t:lapse}, theorem \ref{t:shift} and
theorem \ref{t:mass}.  For the sake of clarity, in this section we
presented the results using the standard 3+1 decomposition. However,
for axially symmetric spacetimes there exists a more natural
decomposition, the (2+1)+1 formalism described in the section
\ref{sec:axisymm-evol-equat}.  The maximal isothermal gauge conditions
and the mass formula are better expressed in this formalism.

\section{Axisymmetric evolution equations}
\label{sec:axisymm-evol-equat}
In this section we review the (2+1)+1 formalism, which has two
parts. First, a reduction of the field equations by the action of the
symmetry \cite{Geroch71} and second a time plus space decomposition of
the reduced equations (see \cite{Nakamura87}, \cite{Choptuik:2003as},
\cite{Rinne:2005df}). We also give some useful formulas which
relate the quantities in the (2+1)+1 formalism with the natural ones
in the standard 3+1 decomposition.
\subsection{The (2+1)+1 formalism}

Consider a vacuum solution of Einstein's equations, i.e., a four
dimensional manifold $\vc$ with metric $\gc_{ab}$ for which the Ricci
tensor $ \rc_{ab}$ vanishes.
Suppose, in addition, that there exists a spacetime
Killing vector  $\eta^a$. We define the norm and the
twist of $\eta^a$, respectively, by
\begin{equation}
  \label{eq:1}
\lambda^2=\eta^a\eta^b g_{ab}, \quad
\Y_a=\epsilon_{abcd}\eta^b \nc^c\eta^d,  
\end{equation}
where $\nc_a$ is the connection  and $\epsilon_{abcd}$ the volume
element with respect to  $\gc_{ab}$. Assuming that the manifold is
simply connected and using
$\rc_{ab}=0$ it is possible to prove that $\omega_a$ is the
gradient of a scalar field $\Y$
\begin{equation}
  \label{eq:2}
 \omega_a=\nc_a \Y. 
\end{equation}
In our case the Killing field will be spacelike, i.e. $\lambda \geq
0$. 

In the presence of a Killing field, there exists a well known
procedure to reduce the field equations \cite{Geroch71}. Let $\vt$
denote the collection of all trajectories of $\eta^a$, and assume
that it is a differential 3-manifold.  We define the Lorentzian metric
$\gt_{ab}$ on $\vt$ by
\begin{equation}
  \label{eq:3}
\gc_{ab}=\gt_{ab}+ \frac{\eta_a\eta_b}{\norm^2}.
\end{equation}
Einstein vacuum equation are equivalent to the following set of
equations intrinsic to $(\vt, \gt_{ab})$ 
\begin{align}
\nt_a \nt^a \norm  & = -\frac{1}{2\norm^3}\nt^a \Y \nt_a \Y \label{eq:hy1s},\\
\nt_a \nt^a \Y & =\frac{3}{\norm}\nt^a\Y \nt_a \norm  \label{eq:hy2s},\\
\rtl_{ab} &= \frac{1}{\norm}\nt_a\nt_b \norm + \frac{1}{2\norm^4}\left( \nt_a \Y
  \nt_b \Y - \nt_c \Y \nt^c \Y h_{ab}\right). \label{eq:ees} 
\end{align}
where $\nt_a$ and $\rtl_{ab}$ are  the connexion and the Ricci tensor
of $\gt_{ab}$.
The system of equations (\ref{eq:hy1s})--(\ref{eq:ees}) can be
interpreted as Einstein equation in 3-dimensions (equation
(\ref{eq:ees})) coupled with a matter sources given by equations
(\ref{eq:hy1s})--(\ref{eq:hy2s}).

We make a $2+1$ decomposition of $(\vt, \gt_{ab})$ and the field
equation (\ref{eq:ees}). Let $\Nt^a$ be the unit normal vector
orthogonal to a spacelike, 2-dimensional slice $\sliced$.  By
construction we have $\Nt^a\eta_a=0$. And hence the norm defined in
(\ref{eq:1}) is equivalent to (\ref{eq:52}).  The intrinsic metric on
$\sliced$ is denoted by $q_{ab}$ and it is given by
\begin{equation}
  \label{eq:6}
h_{ab}=-n_an_b + q_{ab}.
\end{equation}
Our convention for the signature of  $h_{ab}$ is $(-++)$.  In terms
of the lapse $\ld$ and shift vector $\sd^a$ the line element takes the
form
\begin{equation}
  \label{eq:26}
h = -\ld^2dt^2+ q_{ij}(dx^i+\sd^i dt)(dx^j+\sd^j dt),
\end{equation}
where $i,j=1,2$.  The extrinsic curvature $\ed_{ab}$ of the slices
$\sliced$ is given by \footnote{opposite sign convention with respect to
  \cite{Wald84}}
\begin{equation}
  \label{eq:63}
  \ed_{ab}= - q^c_a \nabla_c n_b= -\frac{1}{2}\pounds_n q_{ab}.
\end{equation}
We write the field equations (\ref{eq:ees}) as evolution equations for
$(\qd_{ab},\ed_{ab})$
\begin{align}
  \label{eq:adm1}
\dot \qd_{ab} &= -2\ld\ed_{ab} + \pounds_\sd \qd_{ab},\\
\dot \ed_{ab} &= \pounds_\sd \ed_{ab}+ \ld\left(\ed
  \ed_{ab}+ \rd_{ab}- ^{(3)}{\bar{\mathcal{R}}}_{ab}-2\ed_{ac}\ed^{c}_{b} \right) - D_aD_b\ld,
\label{eq:adm2}
\end{align}
and constraint equations
\begin{align}
  \label{eq:8}
\rd  -\ed^{ab}\ed_{ab}+\ed^2 &=\mu,\\
 D^a\ed_{ab} -D_b\ed &=J_b.  \label{eq:8b}
\end{align}
Here a dot denotes partial derivative with respect to the $t$
coordinate,  $\rd_{ab}$ is the Ricci tensor and $D_a$ is the
connexion with respect to $q_{ab}$,  $^{(3)}{\bar{\mathcal{R}}}_{ab}$
denotes the projection of $^{(3)}{{\mathcal{R}}}_{ab}$ into $S$ and 
\begin{align}
  \label{eq:5}
\mu &= 2\rtl_{ab} n^a n^b+ \rtl,\\
  \label{eq:36}
  J_b & =-q^c_b n^a {\rtl}_{ca}.
\end{align} 
The scalar $\mu$, which appears as `matter source' in the Hamiltonian
constraint (\ref{eq:5}),   will be relevant in the computation of the
mass.  We
will need its explicit expression
\begin{equation}
  \label{eq:7}
\mu =2 \frac{|\Dd\norm|^2}{\norm^2} +2 \Dd_a \left(\frac{\Dd^a\norm}{\norm}
\right)+2 \frac{\norm'}{\norm}\ed + \frac{1}{2\norm^4}\left( \Y'^2 +
  |\Dd\Y|^2\right) 
,
\end{equation}
where the prime denotes directional derivative with respect to $n^a$, that is
\begin{equation}
  \label{eq:14}
\norm' = n^a\nabla_a \norm=\pounds_n \norm 
\end{equation}
Note that
\begin{equation}
  \label{eq:28}
\norm'=\ld^{-1}(\dot\norm -\pounds_\sd\norm).
\end{equation}
To obtain equation  \eqref{eq:7} we have used  equations
(\ref{eq:hy1s})--~(\ref{eq:hy2s}).  

Finally, it convenient to decompose the extrinsic curvature $\ed_{ab}$ in its
trace $\ed$ and trace free part $k_{ab}$, that is 
\begin{equation}
  \label{eq:59}
  k_{ab}=\ed_{ab}-\qd_{ab}\frac{\ed}{2}.
\end{equation}

\subsection{Relations with the 3+1 decomposition}
Take the definition of the spacetime twist potential (\ref{eq:1}) and
(\ref{eq:2}). The derivative $\Y'$ is given by
\begin{equation}
  \label{eq:20}
 \Y'=n^a\nc_a \Y= n^a\epsilon_{abcd }\eta^b \nc^c\eta^d. 
\end{equation}
Note that $ n^a\epsilon_{abcd }$ is the volume element on the
3-dimensional slice $\slicet$ denoted by $\epsilon_{abc}$ in section
\ref{sec:main-result}. The covariant derivative $\nc_a$ in
(\ref{eq:20}) appears antisymetrized and hence we can replace it by
any covariant derivative, in particular $\Dt_a$. Then, we obtain that
(\ref{eq:20}) is equivalent to (\ref{eq:55}).

The twist potential $\Y$ can be calculated in terms of the initial
data as follows. Define the vector $S^a$ by
\begin{equation}
  \label{eq:2t}
    S_a=  K_{ab}\eta^b. 
\end{equation} 
Using the definition of $K_{ab}$
\begin{equation}
  \label{eq:24}
  K_{ab}=-\qt_a^c\nc_c n_b,
\end{equation}
we obtain
\begin{equation}
  \label{eq:50}
  S_a=-n^c\nc_c \eta_a,
\end{equation}
where we have used the Killing equation $\nc_{(a}\eta_{b)}=0$,
equation  $\eta_a n^a=0$ and the following expression for the metric $\qt_{ab}$
\begin{equation}
  \label{eq:84}
  \gc_{ab}=\qt_{ab}- n_a n_b.
\end{equation}

We use the following expression for the covariant derivative of
$\eta_a$ (see \cite{Geroch71})
\begin{equation}
  \label{eq:58}
  \nc_a \eta_b=\frac{1}{2\lambda^2}
  \epsilon_{abcd}\eta^c\Y^d+\frac{1}{2\lambda} \eta_{[a}\Dt_{b]}\lambda. 
\end{equation}
Using (\ref{eq:58}) and (\ref{eq:50}) we obtain
\begin{equation}
  \label{eq:73}
 \epsilon_{abc}  S^b \eta^c=-\frac{1}{2} \Dt_a\Y.
\end{equation}
Note that the left hand side of equation (\ref{eq:73})  is calculated only with
the initial data. This is the desired equation.

The relation between the trace of the extrinsic curvature $\et$ of the slice
$\slicet$ and the trace of extrinsic curvature 
 $\ed$ of the slice $\sliced$ is given by
\begin{equation}
  \label{eq:37}
  K=\ed-\frac{\norm'}{\norm}.
\end{equation} 
If we  prescribe $K=0$ then we have
\begin{equation}
  \label{eq:40}
  \ed=\frac{\norm'}{\norm}. 
\end{equation}
The relation between the two extrinsic curvatures is given by
\begin{equation}
  \label{eq:53}
  \ed_{ab}= q^c_aq^d_b\et_{cd}.
\end{equation}
The following component of the extrinsic curvature will be used in the
next section
\begin{equation}
  \label{eq:54}
\et_{\phi\phi}=\et_{ab}\eta^a\eta^b=  - \norm' \norm. 
\end{equation}
Using equations (\ref{eq:2t}), (\ref{eq:73}), (\ref{eq:53}), and
(\ref{eq:54}) we obtain the following expression for the square of $\et_{ab}$
\begin{equation}
  \label{eq:62}
  \et^{ab}\et_{ab}= \ed^{ab}\ed_{ab}+\frac{\norm'^2}{\norm^2}+
  \frac{|D\Y|^2}{2\norm^4}.  
\end{equation}

\section{Maximal-isothermal gauge and mass integral }
\label{sec:maxim-isoth-gauge}
In this section we review the well known maximal-isothermal gauge (see
\cite{Nakamura87}, \cite{Garfinkle:2000hd} \cite{Choptuik:2003as},
\cite{Rinne:2005df}, \cite{Rinne:2008tk}) using the (2+1)+1 formalism
presented in the previous section. We analyze the fall off behavior of
the fields in this gauge near the axis an at infinity. Finally, we
present a new derivation of the mass integral formula.
\subsection{Gauge}

The maximal gauge condition is well known. In the standard $3+1$
decomposition, the equation for the lapse is given by  
\begin{equation}
  \label{eq:21}
  \Lg \ld = \ld \et^{ab} \et_{ab},
\end{equation}
where $\Lg$ denotes the Laplacian with respect to the metric $\qt_{ab}$. 
 
The shift vector is fixed by the requirement that the isothermal
coordinates are preserved under the evolution.
In isothermal coordinates, the 2-dimensional metric $q_{ab}$ has the
form \eqref{eq:60}. If we take a time
derivative to \eqref{eq:60} we obtain
\begin{equation}
  \label{eq:39}
  \dot q_{ab}= 2\dot u \, q_{ab},
\end{equation}
that is, the tensor $\dot q_{ab}$ is pure trace. Hence, the trace
free part (with respect to $q_{ab}$)  of the
evolution equation \eqref{eq:adm1} is given by
\begin{equation}
  \label{eq:42}
 (\ckq \sd )_{ab}= 2\ld k_{ab},
\end{equation}
where $\ckq$ is the conformal Killing operator with respect to $q_{ab}$ 
\begin{equation}
  \label{eq:43}
 (\ckq \sd )_{ab} =D_a\sd_b + D_b\sd_a -q_{ab}  D_c\sd^c. 
\end{equation}
Equations \eqref{eq:42} constitute a first order elliptic  system 
which determine $\sd^a$ under appropriate boundary conditions, as we
will see in  section \ref{sec:existence-gauge-mass}. Summarizing, the
lapse and shift for the maximal isothermal gauge are determined by
equations (\ref{eq:21}) and (\ref{eq:42}).

\subsection{Fall off behavior and axial regularity in isothermal coordinates}
\label{sec:fall-behavior-axial}

The fall of conditions \eqref{eq:66} for the 3-dimensional metric
implies that (see
\cite{Chrusciel:2007dd})
\begin{equation}
  \label{eq:44}
\u= \x=\oi{-1/2}.
\end{equation}
From equations (\ref{eq:53}) and equation (\ref{eq:76}) we obtain
\begin{equation}
  \label{eq:78}
  \ed_{ij}=\oi{-3/2}.
\end{equation}
The function $\norm'/\norm$ will be important in the next
section, we need to compute its fall off behavior. This behavior
is not an obvious consequence of (\ref{eq:76}) because the relation
between these two functions involves a quotient by $\rho^2$. To
compute the fall off of this function  we
proceed as follows. 
In  Cartesian components we have
\begin{equation}
  \label{eq:79}
  \eta^x=-y, \quad \eta^y=x,
\end{equation}
and hence
\begin{equation}
  \label{eq:67}
 \et_{\phi\phi}= y^2\et_{xx}-2xy\et_{xy}+x^2\et_{yy}.
\end{equation}
Following \cite{Chrusciel:2007dd} it is enough to consider the plane
$x=0$ which is transversal to the Killing vector.  If we use
(\ref{eq:54}), (\ref{eq:67}) and (\ref{eq:61}) and then evaluate at
$x=0$ we have
\begin{equation}
  \label{eq:80}
  -\frac{\norm'}{\norm}= \frac{ \et_{\phi\phi}}{\norm^2}=e^{-\x}\et_{xx}.
\end{equation}
Hence, at this plane the function $\norm'/\norm$ is smooth and it has
the same fall off as $K_{ab}$, namely
\begin{equation}
  \label{eq:23}
 \frac{ \norm'}{\norm}= \oi{-3/2}.
\end{equation}
Since the plane $x=0$  is transversal to the axial Killing vector it follows
that (\ref{eq:23}) holds everywhere.

The regularity of the metric $\qt_{ab}$ and the extrinsic curvature
$\et_{ab}$ at the axis implies restrictions on the behavior of the
different quantities (for a detailed discussion of this issue see
\cite{Rinne:2005df} and \cite{Rinne:2005sk}). We summarize here two
important consequences of axial regularity. First, the function $\q$
defined by
\begin{equation}
  \label{eq:31}
  \q=\u -\frac{\x}{2},
\end{equation}
vanished at the axis (see
\cite{Chrusciel:2007dd}) 
\begin{equation}
  \label{eq:q}
  \q(\rho=0)=0.
\end{equation}
Second, the functions involved satisfy parity conditions with
respect to the $\rho$ coordinate. In particular we have
\begin{equation}
  \label{eq:10}
  \u, \q, \x, k_{\rho}^\rho \text{ are even functions of }\rho,
\end{equation}
and
\begin{equation}
  \label{eq:9}
  \norm, k^\rho_z \text{ are odd functions of }\rho.
\end{equation}
Note that odd functions vanishes  at the axis.

\subsection{Mass formula}
Consider the Hamiltonian constraint  \eqref{eq:8}. Using equations
\eqref{eq:61} and (\ref{eq:31}) we write the Hamiltonian constraint in
terms of $\rho$, $\x$ and $\q$ in such a way that there are no
singular terms at the axis.  To do this, we use the following
elementary identities. 

The Ricci scalar of $q_{ab}$ is given by
\begin{equation}
  \label{eq:ricciscalar}
\rd=-2e^{-2u}\Ld u,
\end{equation}
where $\Ld$ is the flat 2-dimensional Laplacian
\begin{equation}
  \label{eq:70}
 \Ld u=  u_{,\rho\rho}+ u_{,zz}.
\end{equation}
Using equation (\ref{eq:61}) we obtain 
\begin{equation}
  \label{eq:16}
\frac{|D\norm|^2}{\norm^2}= |D\log \norm|^2=|D(\frac{\x}{2} +\log \rho)|^2.
\end{equation}
Finally, the $\rho$ coordinate satisfy 
\begin{equation}
  \label{eq:19}
\Ldt \log \rho  =0.
\end{equation}
where we have defined  $\Ldt$ as  the 3-dimensional flat Laplacian 
\begin{equation}
  \label{eq:20c}
\Ldt  u = \Ld u + \frac{u_{,\rho}}{\rho}.
\end{equation}

Using equations \eqref{eq:31}, (\ref{eq:ricciscalar}), (\ref{eq:16}) and
(\ref{eq:19}) the Hamiltonian constraint \eqref{eq:8} can be written in the
following form
\begin{equation}
  \label{eq:46}
 -2\Ldt \x -2\Ld \q =
\frac{e^{2u}}{2}\left[
  2\ed^{ab}\ed_{ab}+2\ed(-\ed+2\frac{\norm'}{\norm}) + |D\x |^2 
+ \frac{1}{\norm^4}\left( \Y'^2 + |D\Y|^2 \right)\right].  
\end{equation}
We define the mass as the integral of the right  hand side of this
equation, namely
\begin{equation}
  \label{eq:114}
m=\frac{1}{16}\int_{S}\left(2k^{ab}k_{ab}-\ed^2+4\chi
  \frac{\norm'}{\norm}  + |D\x |^2  
+ \frac{1}{\norm^4}\left( \Y'^2 + |D\Y|^2 \right)   \right )\rho \dvq,   
\end{equation}
where we have used (\ref{eq:59}) and $\dvq=e^{2u}d\rho dz$ denote
the volume element with respect to $q_{ab}$. Note that we introduce a
weight $\rho$ in the integral. If we impose the  condition $K=0$
(using equation (\ref{eq:37})) we get that the integrand is positive
definite, namely
\begin{equation}
  \label{eq:15b}
m= \frac{1}{16}\int_{S}\left(2k^{ab}k_{ab}+3\frac{\norm'^2}{\norm^2}  + |D\x |^2
+ \frac{1}{\norm^4}\left( \Y'^2 + |D\Y|^2 \right)   \right )\rho \dvq. 
\end{equation}
For maximal 2-dimensional slices $\ed=0$ the integral
(\ref{eq:114}) is also positive definite.  However, this choice will
not give an equation for the lapse function (analog to (\ref{eq:21}))
with positive definite solutions.  
To prove the equivalence of equation (\ref{eq:15b}) with (\ref{eq:15})
we use (\ref{eq:62}).

The crucial point is that this volume integral can be written as a
boundary integral at infinity using the left hand side of equation
(\ref{eq:46}). We have
\begin{equation}
  \label{eq:17}
  m=-\frac{1}{4}\int_{-\infty}^{\infty}dz \int_0^\infty \left( \Ldt \x +\Ld
    \q\right)  \rho\, d\rho.  
\end{equation}
This integral can be converted into a boundary integral at infinity as
follows. For the first term we use that $\rho\, d\rho dz$ is the
volume element in $\Rt$ and then we can use the divergence theorem in
three dimension to obtain 
\begin{equation}
  \label{eq:32}
\int_{-\infty}^{\infty}dz \int_0^\infty   \Ldt \x  \rho d\rho dz=
\lim_{r\to\infty}\int_0^\pi \partial_r \, \x r \rho d \theta. 
\end{equation}
The second  term in (\ref{eq:17}) can be also written in divergence
form 
\begin{equation}
  \label{eq:86}
 \int d \rho 
 \, d z \, (q_{,\rho \rho }+q_{,z z})  \rho =  \int d \rho \, d z \, \left(
    (\rho q_{,\rho} - q)_{,\rho }+ (\rho q_{,z})_{,z} \right).
\end{equation}
We use the divergence theorem in two dimensions to transform this
volume integral in a boundary integral. Namely, let $\Omega$ be an
arbitrary domain, we have 
\begin{equation}
  \label{eq:87}
\int_{\Omega} d \rho \, d z \, \left(
    (\rho q_{,\rho} - q)_{,\rho }+ (\rho q_{,z})_{,z} \right) =
  \oint_{\partial \Omega} 
  \bar V\cdot \bar n \, d \bar s,  
\end{equation}
where $\bar n$ is the 2-dimensional unit normal, $ d \bar s$ the line
element of the 1-dimensional boundary $\partial \Omega$ and $\bar V$
is the 2-dimensional vector given in coordinates $(\rho, z)$ by
\begin{equation}
  \label{eq:88b}
\bar V= ((\rho q_{,\rho} - q), (\rho q_{,z}) ).
\end{equation}
Let $\Omega$ be the half plane $\rho\geq 0$. 
By \eqref{eq:q} and the assumption that $q$ is smooth  we have
that the vector $\bar V$ vanishes at the axis.  Then  the only contribution
of the boundary integral is at infinity, namely
\begin{equation}
  \label{eq:97}
 \int_{\rho\geq 0} \Delta_2 q \rho \,d\rho dz  =\lim_{r\to \infty}\int_0^\pi 
 \left(r \partial_r
 q-q\right)\rho \, d\theta .
\end{equation} 
Summing (\ref{eq:97}) and (\ref{eq:32}) we obtain 
 the final expression for the mass as boundary integral
\begin{equation}
  \label{eq:massb}
  m=-\frac{1}{4} \lim_{r\to \infty} \int_0^\pi  \left(\partial_r \x
    + \partial_r q-\frac{q}{r}\right) \, r\rho \, d\theta.       
\end{equation}
We can write this expression in terms of $u$ instead of $q$
\begin{equation}
  \label{eq:massbu}
  m=-\frac{1}{4} \lim_{r\to \infty}\int_0^\pi \left(\frac{1}{2}\partial_r \x +
    \frac{\x}{2r} + \partial_r 
    u-\frac{u}{r}\right)   \, r\rho \, d\theta.   
\end{equation}
Up to now we have not mention the relation of $m$ with the total mass.
In fact in order to prove theorem \ref{teo:main-result-1} we do not
need any other information. The quantity $m$ is a positive definite
integral that can be expressed as a boundary integral at infinity. As
a consequence we will prove that it is conserved along the evolution
in the maximal-isothermal gauge. 

However, it is possible to show that $m$ is precisely the total mass
of the spacetime. There is an interesting and subtle point here. The
equivalence of this formula with the total mass have been previously
proved using a strong fall off condition for $q$ (see \cite{Brill59},
\cite{Gibbons06}, \cite{Dain06c}), namely
\begin{equation}
  \label{eq:48}
  q=\oi{-3/2},
\end{equation}
which implies that $q$ does not appears in the boundary integral
(\ref{eq:massbu}). Recently, this equivalence was proved without this
restriction \cite{Chrusciel:2007dd}. In the deduction we have made, we
have not use any restriction for $q$ which is consistent with
\cite{Chrusciel:2007dd}. Note the simplicity of this deduction (which
is much in the spirit of the energy defined in \cite{Choquet-Bruhat01}
for cosmologies) compared with the ones using three dimensional slices like
\cite{Gibbons06}. 

We also mention that in the next section we will prove that the strong
decay condition (\ref{eq:48}) is preserved by the evolution.

\section{Existence of the gauge and mass conservation}
\label{sec:existence-gauge-mass}
We begin with the lapse equation \eqref{eq:21}.  This condition does
not require any symmetry on the data. The existence and uniqueness of
the solution of this equation in the standard $3+1$
decomposition are well known. For completeness we review this result.

\begin{lemma}\label{t:lapse}
  There exists a unique smooth strictly positive solution $\ld$ of  equation
  \eqref{eq:21} with the
  following fall off
  \begin{equation}
    \label{eq:65}
    \ld -1= \oi{-1/2}.
  \end{equation}
\end{lemma}

\begin{proof} 
Set
\begin{equation}
  \label{eq:4}
 \bar\ld =\ld-1,
\end{equation}
then, using \eqref{eq:21}, we have that 
\begin{equation}
  \label{eq:57}
 \Lg \bar\ld   -f \bar \ld=f,
\end{equation}
where $f=K^{ab}K_{ab}\geq 0$. The fall off behavior assumed for
$K_{ab}$ in \eqref{eq:76} implies that $f\in \hkb{s}{-3/2}$ for all
$s$.  Here $\hkb{s}{\delta}$ denote weighted Sobolev spaces.  See
\cite{Bartnik86} \cite{Cantor81} \cite{Choquet81} and reference
therein for definition and properties of these spaces. We use the
index notation of \cite{Bartnik86}.

The decay assumed on the metric $\qt_{ab}$ \eqref{eq:66} implies that
the operator $\Lg-f:H^{s+2}_\delta\to H^s_{\delta-2} $ is an
isomorphism for $\delta=-1/2$ (see for example \cite{Choquet99},
\cite{Maxwell04}). Note that we use this theorem for dimension 3.
Then there exists a unique solution $\bar\ld\in \hkb{s}{-1/2} $ of
equation \eqref{eq:57}.  

Since $\bar \alpha \to 0$ as $r\to \infty$ (because $\bar \alpha \in
H^{s}_{-1/2}$), it follows that $\alpha \to 1$ as $r\to \infty$. Hence, we can
use the version of the maximum principle for non-compact manifolds presented in
\cite{Choquet99} to conclude that there exists $\epsilon>0$ such that $\alpha
\geq \epsilon $.

\end{proof}

If the data is axially symmetric then the solution $\ld$ will be also
axially symmetric: take the Lie derivative with respect to $\eta^a$ to
both sides of equation \eqref{eq:57}, we get
\begin{equation}
  \label{eq:68}
  \Delta_\qt\left( \pounds_\eta\bar \ld\right)-f\left(
    \pounds_\eta\bar \ld\right) =0,
\end{equation}
since $ \pounds_\eta\bar \ld\in\hkb{s}{-1/2} $ (to see this, take the
expression of $\eta$ in Cartesian coordinates (\ref{eq:79})), the
isomorphism theorem mentioned above implies that $\pounds_\eta\bar
\ld=0$. Finally, we note that since $\ld$ is smooth in $(x,y,z)$, it
follows that an axially symmetric $\ld$ is an even function of $\rho$.

We consider now the equation for the shift.  Instead of solving the
first order system \eqref{eq:42} we will solve a second order system
obtained from taking a derivative to \eqref{eq:42}. Then, we will show
that a solution of the second order system is also a solution of the
original equations \eqref{eq:42} under appropriate boundary
conditions.

Under the conformal transformations \eqref{eq:60}, the conformal
Killing operator rescale like
\begin{equation}
  \label{eq:64}
 \ckq(\sd)_{ab}=e^{2u}\ck(e^{-2u} \sd)_{ab},
\end{equation}
where $\ck$ is the flat conformal Killing operator
\begin{equation}
  \label{eq:69}
(\ck \sd )_{ab} =\partial_a\sd_b + \partial_b\sd_a -q_{ab}  \partial_c\sd^c.  
\end{equation}
We take  a flat divergence to equation
~\eqref{eq:42} and use equation \eqref{eq:64} to get 
\begin{equation}
  \label{eq:shift-inter}
  \Ld \left( e^{-2u} \sd_a\right) =2\partial^b (\ld e^{-2u} k_{ab}), 
\end{equation}
where we have used
\begin{equation}
  \label{eq:71}
\partial^b (\ck \sd )_{ab} =\Ld \sd_a. 
\end{equation}
If we contract equation \eqref{eq:shift-inter} with  $\delta^{ab}$ and use
$\qd^{ab}=e^{-2u}\delta^{ab}$ we obtain our final equation
\begin{equation}
  \label{eq:shift}
  \Ld \sd^a =2\partial^b (\ld k_{b}^a),  
\end{equation}
where $\sd^a=\qd^{ab}\sd_b$ and $ k_{b}^a=\qd^{ac}k_{cb}$

As it was mentioned in section \ref{sec:fall-behavior-axial}, the
smoothness of the metric at the axis implies that the relevant
functions satisfies parity conditions in the $\rho$ coordinate.  In
particular, the components of the shift vector $\sd^a$ should satisfy
\begin{equation}
  \label{eq:88}
  \sd^z(\rho,z)=\sd^z(-\rho,z)  , \quad \sd^\rho(\rho,z)=
  -\sd^\rho(-\rho,z). 
\end{equation}

\begin{theorem}\label{t:shift}
  There exists a unique, smooth,  solution $\sd^a$   of equation
  \eqref{eq:shift} with the following fall off behavior
  \begin{equation}
    \label{eq:11}
    \sd^a=\oi{-1/2}.
\end{equation}
The solution  satisfies the parity conditions \eqref{eq:88} and the
gauge equation \eqref{eq:42}.  

Moreover, the quotient $\sd^\rho /\rho$ is smooth and have the
following fall off
\begin{equation}
  \label{eq:92}
   \frac{ \sd^\rho}{\rho}=\oi{-3/2}.
\end{equation}

\end{theorem}
Remark: the fall-off condition (\ref{eq:92}) will be essential in the
proof of the mass conservation theorem \ref{t:mass}. 
 
\begin{proof}
  We distinguish the two components of equation \eqref{eq:shift}
  \begin{equation}
    \label{eq:27}
     \Ld \sd^\rho =F^\rho ,\quad  \Ld \sd^z =F^z,
  \end{equation}
where
\begin{equation}
  \label{eq:29}
  F^\rho =2\partial^j (\ld k_{j}^\rho) , \quad   F^z =2\partial^j (\ld
  k_{j}^z). 
\end{equation}
The function $\ld$ is even in $\rho$, using (\ref{eq:10}) and
(\ref{eq:9}) we obtain that  $ F^\rho$ is odd and $F^z$ is even in
$\rho$. 

In terms of the cylindrical coordinates $(\rho,z)$, the source functions
$F$ are defined only for $\rho\geq 0$. But because they satisfy   parity
conditions we can smoothly extend them to all  
$\Rd$  by the following prescription
\begin{equation}
  \label{eq:30}
 F^\rho (\rho,z)=  -F^\rho(-\rho,z), \quad  F^z(\rho,z)=  F^z(-\rho,z).
\end{equation}
Hence, we can consider equations \eqref{eq:27} as Poisson equations
\begin{equation}
  \label{eq:95}
  \Ld u=f,
\end{equation}
in $\Rd$ with decay conditions at infinity  (\ref{eq:11}).
Uniqueness of solution of (\ref{eq:95}) under the decay condition
(\ref{eq:11}) follows immediately integrating by parts the Laplace
equation. Namely, for any domain $\Omega$ we have
\begin{equation}
  \label{eq:93}
\int_{\Omega} u \Ld u=\int_{\Omega} \partial (u\partial u)-|\partial
u|^2=\oint_{\partial\Omega} u\partial_n u -\int_{\Omega}|\partial u|^2, 
\end{equation}
and hence for a solution of $\Ld u=0$ we have
\begin{equation}
  \label{eq:94}
 \oint_{\partial\Omega} u\partial_n u =\int_\Omega |\partial u|^2. 
\end{equation}
Let $\Omega$ be a ball of radius $R$, and take the limit $R\to
\infty$. By the decay condition (\ref{eq:11}),  the boundary integral
in (\ref{eq:94}) vanishes. Then, it follows that $u$ must be constant, by
the decay condition it must be zero. Since the homogeneous equation
has only the trivial solution uniqueness follows. 
Also, for solutions in this
class, integrating  equation (\ref{eq:95}) we obtain that the source
must satisfy the condition 
\begin{equation}
  \label{eq:74}
  \int_{\Rd} f =0.
\end{equation}
 We need to verify \eqref{eq:74} for our source functions.  The
 function  $F^\rho$ is odd in $\rho$, hence it satisfies \eqref{eq:74}
automatically.  For $F^z$
we use that it can be written in divergence form
\begin{equation}
  \label{eq:33}
  F^z=\partial_i v^i, \quad v^i=2 \ld k^{i}_z .
\end{equation}
Then we have
\begin{equation}
  \label{eq:34}
   \int_{\Rd}  F^z  =2\int_{\rho\geq 0}  F^z =2\int_{\rho\geq 0} \partial_i
   v^i
\end{equation}
using the divergence theorem in 2-dimension we get
\begin{equation}
  \label{eq:72}
 \int_{\rho\geq 0} \partial_i
   v^i= \int_{0}^\infty v^\rho|_{\rho=0}\, d\rho + \lim_{R\to
     \infty}\oint_{C_R}v^r ds, 
\end{equation}
where $C_R$ denotes the semicircle with $r=R$. 
By the decay condition (\ref{eq:78}) for $k_{ab}$ and (\ref{eq:65})
for the lapse, we have that the second boundary
integral vanishes. For the first one, we  equation \eqref{eq:9} to
conclude that
\begin{equation}
  \label{eq:35}
   v^\rho|_{\rho=0}=0.
\end{equation}
And hence we obtain the desired result that $F^z$ satisfies
(\ref{eq:74}). 

To actually obtain a solution of equation (\ref{eq:95}) we use the
Green formula for the Laplacian in two dimension
$G(x,x')=-\log|x-x'|^2$.  A solution $u$ has the integral
representation
\begin{equation}
  \label{eq:81}
  u(x)=\int_{\Rd} G(x,x') f(x') dx'. 
\end{equation}
In polar coordinates $(r,\theta)$ (related to ($\rho, z$) by the
standard formula $z=r\cos\theta$ $\rho=r\sin\theta$) we have the well
known expansion of $G$
\begin{align}
  \label{eq:96}
  G(r,\theta;r',\theta') & =-2 \log r + 2\sum_{n=1}^{\infty}\frac{r'^n}{n
    r^n}\cos(n(\theta-\theta') )\text{ for }r>r', \\
& =-2 \log r' + 2\sum_{n=1}^{\infty}\frac{r^n}{n
    r'^n}\cos(n(\theta-\theta') )\text{ for }r<r'.
\end{align}
We use the trigonometric identity
\begin{equation}
  \label{eq:98}
  \cos(n(\theta-\theta') )= \cos(n\theta) \cos(n\theta')+
  \sin(n\theta) \sin(n\theta'),
\end{equation}
to split the sum in terms which are odd functions of $\theta$ (and
hence of $\rho$) and terms which are even functions of $\theta$. 

If the source function $f$ in the integral (\ref{eq:81}) is odd in
$\rho$, then the terms involving $\cos(n\theta')$ vanish, for even
source functions the terms with $\sin(n\theta')$ vanish.

We are interested in the fall off behavior of the solution. It is
convenient to split the source functions $F$ in two term, one with
compact support in some ball and the other which pick up the decay
behavior. Let $\chi:\mathbb{R}\to \mathbb{R}$ be a cut off function
such that $\chi \in C^\infty(\mathbb{R})$, $0\leq\chi\leq 1$,
$\chi(t)=1$ for $0\leq t \leq 1$, $\chi(t)=0$ for $2\leq t$ and write
$\chi_R(r)=\chi(r/R)$, $f_\infty=(1-\chi_R) f$, where $R$ is an
arbitrary positive number.

Using the expansion (\ref{eq:96}) we compute the integral
(\ref{eq:81}). Note that the term with $\log r$  vanish because
(\ref{eq:74}). 
For $r>R$ we have for $\sd^\rho$
\begin{equation}
  \label{eq:99}
  \sd^\rho=  2\sum_{n=1}^{\infty}\frac{\sin(n\theta)}{n} \left[
    r^{-n}A_n+B_n(r) + C_n(r)) \right], 
\end{equation}
where
\begin{align}
  \label{eq:102}
  A_n &=\int_0^R r'^n \sin(n\theta') F^\rho r' dr'd\theta', \\
  \label{eq:13}
  B_n(r) &=r^{-n}\int_{R}^r r'^n \sin(n\theta') F_\infty^\rho \,
  r'  dr'd\theta',\\
  C_n(r) & =r^n \int_{r}^\infty \frac{1}{r'^n} \sin(n\theta')
  F_\infty^\rho \, r'  dr'd\theta'. 
\end{align}
For $\sd^z$ we have a similar expansion for  $r>R$
\begin{equation}
  \label{eq:99z}
  \sd^z=  2\sum_{n=1}^{\infty}\frac{\cos(n\theta)}{n} \left[
    r^{-n}A_n+B_n(r) + C_n(r)) \right] ,
\end{equation}
where
\begin{align}
  \label{eq:102z}
  A_n &=\int_0^R r'^n \cos(n\theta') F^z r' dr'd\theta' , \\
  \label{eq:13z}
  B_n(r) &=r^{-n}\int_{R}^r r'^n \cos(n\theta') F_\infty^z \,
 r'   dr'd\theta',\\
  C_n(r) & =r^n \int_{r}^\infty \frac{1}{r'^n} \cos(n\theta')
  F_\infty^z \, r'  dr'd\theta'. 
\end{align}
We use the decay condition (\ref{eq:78}) to obtain that the sources
satisfies
\begin{equation}
  \label{eq:106}
  F^\rho=F^z=\oi{-5/2}.
\end{equation}
From (\ref{eq:106}) we get that
\begin{equation}
  \label{eq:108}
  B_n(r)=  C_n(r)=\oi{-1/2}.
\end{equation}
These are the solutions of our problem. It is clear that they
satisfy the parity conditions \eqref{eq:88}. 

The important property of the representation (\ref{eq:99}) is
that it allow us to prove (\ref{eq:92}). In effect, we use the
trigonometric identity \cite{Gradshteyn80}
\begin{equation}
  \label{eq:105}
  \sin(n\theta)= \sum_{k=0}^{n-1} \binom{n}{k}   (\cos\theta)^k (\sin\theta
  )^{n-k}\sin\left(\frac{1}{2}(n-k)\pi \right)
\end{equation}
to obtain that the function $\Theta_n$ defined by 
\begin{equation}
  \label{eq:12}
\Theta_n(\theta)= \frac{ \sin(n\theta)}{\sin\theta}= \sum_{k=0}^{n-1}
 \binom{n}{k}
 (\cos\theta)^k (\sin\theta)^{n-k-1} 
  \sin\left(\frac{1}{2}(n-k)\pi \right) 
\end{equation}
is an smooth function. Hence, using (\ref{eq:99}),  we obtain the expansion 
\begin{equation}
  \label{eq:104}
 \frac{\sd^\rho}{\rho}= 2\sum_{n=1}^{\infty}\frac{\Theta_n(\theta)}{n} \left[
    r^{-n+1}A_n+r^{-1}B_n(r) + r^{-1}C_n(r)) \right].  
\end{equation}
Using a similar argument as above we obtain (\ref{eq:92}).

Finally, we need to verify that our solution of equation
\eqref{eq:shift} is also a solution of the first order system
\eqref{eq:42}. Define the trace free tensor $t_{ab}$ by
\begin{equation}
  \label{eq:77}
  t_{ab}=\left( (\ckq \sd )_{ab}- 2\ld k_{ab}\right)e^{-2u}.
\end{equation}
We have proved that 
\begin{equation}
\label{eq:dtd}
\partial^a t_{ab}=0.
\end{equation} 
Let $t_{\rho\rho}=t_1$ and $t_{\rho z}=t_2$, then equations
\eqref{eq:dtd} are given by
\begin{equation}
 \label{eq:cr}
\partial_1 t_1 + \partial_2 t_2=0,\quad \partial_1 t_2 - \partial_2 t_1=0.
\end{equation}
These are the Cauchy-Riemann equation for a complex function
$f=t_1+it_2$. That is, a non trivial solution of \eqref{eq:cr} implies
that $f$ is an entire function on the complex plane. But $t_1$ and
$t_2$ decay to zero at infinity. Hence, by Liouville's theorem, $f$
vanishes. And then, we obtain that $\sd$ is also a solution of the gauge
equation (\ref{eq:42}).
\end{proof}

Solutions of equation \eqref{eq:95} can be obtained by other methods,
for example see Proposition 2.6 in \cite{Chrusciel:2007dd}. In the
proof above we have used the Green function in order to prove
\eqref{eq:92}. It is likely that also this kind of expansions can be
obtained by other methods like the ones used in \cite{chrusciel90}.

To prove the mass conservation we will make use of the evolution
equation given by the trace of (\ref{eq:adm1}), namely
\begin{equation}
  \label{eq:41}
 2 \dot u =-\ld\ed+D_a\sd^a.
\end{equation}
In terms of partial derivatives this equation is written like
\begin{equation}
  \label{eq:22}
  2 \dot u =-\ld\ed+\partial_i\sd^i +2\sd^i\partial_i u.
\end{equation}

\begin{theorem}[Mass conservation]
\label{t:mass}
We have
\begin{equation}
  \label{eq:25}
  \frac{dm(t)}{dt}=0.
\end{equation}
  
\end{theorem}
\begin{proof}
  We use equation \eqref{eq:massbu} to get 
  \begin{equation}
    \label{eq:49}
    \frac{dm(t)}{dt}=\frac{1}{16} \lim_{r\to \infty}\int_0^\pi
    \left(\frac{1}{2}\partial_r \dot \x  +
    \frac{\dot \x}{2r} + \partial_r 
   \dot u-\frac{\dot u}{r}\right)   \, r\rho \, d\theta.   
  \end{equation}
  In order to prove that this integral is zero we have to compute the
  decay of the integrand at infinity. We begin with the terms containing $\dot
  \x$. 

We use the relation
\begin{equation}
  \label{eq:45}
  \frac{ \norm'}{\norm}=\frac{\x'}{2} +\frac{\rho'}{\rho},
\end{equation}
and 
\begin{equation}
  \label{eq:47}
  \rho' =-\frac{\sd^\rho}{\ld},\quad \x'=\frac{1}{\alpha} \left(\dot
    \x-\sd^i\partial_i \x \right),
\end{equation}
to conclude that
\begin{equation}
  \label{eq:83}
  \dot \x=2\ld \frac{
    \norm'}{\norm}+\sd^i\partial_i \x +2\frac{\sd^\rho}{\rho}.
\end{equation}
Using equations (\ref{eq:23}) and (\ref{eq:92}) we obtain
\begin{equation}
  \label{eq:109}
  \dot \x=\oi{-3/2}.
\end{equation}
For $\dot \u$ we use equation (\ref{eq:22}) and the decay on $\sd$ to get
\begin{equation}
  \label{eq:75}
   \dot\u =\oi{-3/2}.
\end{equation}
Using (\ref{eq:109}) and (\ref{eq:75}) we obtain that the boundary
integral \eqref{eq:49} vanishes.

\end{proof}

We mention also another consequence of equation (\ref{eq:22}). 
From this equation we deduce that
\begin{equation}
  \label{eq:82}
  \dot \q =\oi{-3/2}.
\end{equation}
Then, if the function $\q$ has initially the stronger decay
(\ref{eq:48}),  this decay will be preserved by
the evolution.

\section{Final comments}
\label{sec:final-comments}
In this final section we discuss the implication of the results
presented here for the evolution of  axially symmetric
isolated systems.


The first question one need to face in the evolution problem is the
choice of gauge. In axial symmetry there have been studied different
kind of gauge conditions (see for example \cite{Bardeen83}
\cite{Rinne:2005sk}).  The mass conservation formula single out a
particular one. If we define appropriate Sobolev norms, then the
mass formula essentially implies that the $H^1$ norm of the metric and
the $L^2$ norm of the extrinsic curvature are bounded along the
evolution. This is very desirable property which is not present in
other gauges. The mass conservation formula strongly suggests that the
most convenient gauge for the axially symmetric evolution problem is
the maximal isothermal one.

As it was mention in the introduction, the conservation of mass is
closely related with the Hamiltonian formulation of General
Relativity. The maximal isothermal gauge is a dynamical gauge that
depends on the time coordinate. The analysis of time dependent gauge
conditions and its relation with conserved quantities in the
Hamiltonian formulation was recently studied in \cite{Szabados06}
\cite{Szabados03}.  The maximal isothermal gauge satisfies the fall
off conditions described in these references.

In axial symmetry, even when the gauge conditions are fixed there
exists many possibilities to extract from Einstein equations a set of
evolution equations. The reason for this ambiguity is that Einstein
equations in 3-dimensions has no dynamics, this essentially means that
equation (\ref{eq:ees}) can be replaced by the constraint equations
(\ref{eq:8})--(\ref{eq:8b}) (see \cite{Choquet-Bruhat01}). If we chose
to do so, then we obtain an evolution scheme in which the evolution
equations are given by (\ref{eq:hy1s})--(\ref{eq:hy2s}), and the other
equations (including gauge conditions) are elliptic constraint.  But
there exists others alternatives. Following \cite{Rinne:2005df}, we
can classify them by the number of evolutions equations used vs the
number of elliptic constraint equations. The scheme presented above
has the minimum number of evolution equations and the maximum of
constraint equations. This is called a fully constrained system. On
the opposite side we have a system in which we do not solve for the
constrains and we solve the evolution equation
(\ref{eq:adm1})--~(\ref{eq:adm2}). This is called a free evolution
scheme.  In between we have other possibilities to construct partially
constrained systems.  In all cases, the gauge conditions are solved as
elliptic constraints. That is all schemes have a mix of evolution and
elliptic equations (see the discussion in chapter 3 of
\cite{Rinne:2005df}).

To prove the mass conservation we have used two equations: the
Hamiltonian constraint (\ref{eq:8}) and the evolution equation for the
conformal factor of the 2-dimensional metric (\ref{eq:41}). In a free
evolution scheme the Hamiltonian constraint is not solved. Hence, the
mass will not be conserved for arbitrary data. It will be only
conserved for data that satisfies the constraint equations. Then, in
this case the mass integral formula will be not useful for controlling
the evolution. The same will happen with partially constrained
systems in which the Hamiltonian constrain is not solved or the
evolution equation (\ref{eq:41}) is not used. Such systems was
discussed in \cite{Rinne:2005df} and in \cite{Choptuik:2003as} (in
this reference, this kind of systems appears when the evolution of the
conformal factor is used instead of the Hamiltonian constraint).

On the other hand, for a fully constrained scheme (like the ones used
in \cite{Choptuik:2003as}, \cite{Rinne:2008tk}) and a partially
constrained scheme in which both the Hamiltonian and the evolution
equation for the metric are used (like the one studied in
\cite{Garfinkle:2000hd}) the mass will be conserved. The mass
conservation formula single out these two schemes. The natural
problem now is to study their well possness.

Finally, it is interesting to note that the integral mass formula has
direct application in numerical simulations.  In numerical simulations
of isolated systems, one often uses a position-dependent resolution
that is high in the central region and much lower close to the outer
boundary of the computational domain. Hence computing the mass as a
volume integral is more accurate than as a boundary integral.  The
resulting approximation to the mass is much better conserved during
the evolution\footnote{I thank O. Rinne for pointing this out to me}
\cite{Rinnepriv}.

\section*{Acknowledgments}
It is a pleasure to thank Piotr Chru\'sciel, Robert Geroch and Martin
Reiris for discussions. Special thanks to Oliver Rinne for discussions
and for sharing his numerical results on a comparison of the two
different ways of computing the mass.

The author is supported by CONICET (Argentina).
This work was supported in part by grant PIP 6354/05 of CONICET
(Argentina), grant 05/B270 of Secyt-UNC (Argentina) and the Partner
Group grant of the  Max Planck Institute for Gravitational Physics,
Albert-Einstein-Institute (Germany).


\end{document}